\documentclass[12pt]{article}

\usepackage[margin=1in]{geometry}  
\usepackage{graphicx}              
\usepackage{amsmath}               
\usepackage{amsfonts}              
\usepackage{amsthm}                
\usepackage{tikz}               
\usetikzlibrary{arrows,positioning}
\usetikzlibrary{calc}

\newtheorem{thm}{Theorem}[section]
\newtheorem{lem}[thm]{Lemma}

\DeclareMathAlphabet{\mathbmit}{OML}{cmm}{b}{it}
\renewcommand{\vec}[1]{\mathbmit{#1}}
\let\matr\vec

\begin{document}
\date{}

\nocite{*}

\title{A note on linear fractional set packing problem }

\author{Pooja Pandey \thanks{ Corresponding author. Email: poojap@sfu.ca}\\
Department of Mathematics, Simon Fraser University\\
 250 - 13450 – 102nd Avenue, Surrey, BC, V3T 0A3, Canada}

\maketitle

\begin{abstract}
  In this note we point out various errors in the paper by Rashmi Gupta  and R. R. Saxena, {\it Set packing problem with linear fractional objective function}, International Journal of Mathematics and Computer Applications Research (IJMCAR), 4 (2014) 9 - 18. We also provide some additional results.
\end{abstract}



\section{Introduction}
The set packing problem,  set covering problem, and set partitioning problem are among the most well-studied problems in combinatorial optimization and they have wide range of real life applications ~\cite{ aro,aro1, baz, bec, gar, lem}.   Arora, Puri and Swarup ~\cite{aro,aro1} developed couple of solution algorithms for the set covering problem with the linear fractional objective function, where they  exploited the structural properties  of the set covering problems. Later, Gupta and Saxena ~\cite{gupta1} extended these results for the set packing problems with the linear fractional objective function. \\

In this note, we show that the properties established in ~\cite{gupta1} are incorrect. Gupta and Saxena ~\cite{gupta} extended results of~\cite{aro,aro1} to the linear fractional set packing  problems, but these extensions  suffer many drawbacks since they overlooked the structural properties of the set packing problem and ignore the required conditions  for their results to be correct.

\section{The linear fractional set packing problem} \label{lfspp}

Let $\mathcal{E} = \{1,2,\dotsc,m\}$ be a finite set and  $\mathcal{F} = \{S_1,S_2,\dots,S_n\}$  be a family of subsets of $\mathcal{E}$. The index set for elements of $\mathcal{F}$ is denoted by $G= \{1,2 ,\dotsc, n\} $. For each element $j \in G$, a cost $c_j$ and a weight $d_j$  are prescribed. We refer to $c_j$ as the  {\it linear cost } of the set $S_j$ and $\vec{c} = (c_1,\dotsc,c_n)$ as the {\it linear cost vector}. Similarly $d_{j}$ is referred to as the {\it linear weight} of the set $S_j$ and $\vec{d} = (d_1,\dotsc,d_n)$ as the {\it linear weight vector}. $\alpha$ and $\beta $ are constants where  $\beta > 0$.\\

A subset $H$ of $G$ is said to be a {\it pack } of $\mathcal{E}$ if $\bigcup_{j \in H} S_j = \mathcal{E}$, and  $j, k \in H$, $j \neq k$, implies $S_j \bigcap S_k = \emptyset.$ \\

Then the {\it linear set packing problem} (LSPP) is to select a {\it pack } $ H=\{\pi(1), \dotsc, \pi(h)\} $ such that  $\sum_{i=1}^{h} c_{\pi(i)}$ is maximized. Likewise the {\it linear fractional set packing problem} (LFSPP) is to select a { \it pack } $H= \{\sigma(1), \dotsc,\sigma(h)\}$ such that $\dfrac{\sum_{i=1}^{h} \vec{c}_{\sigma(i)} + \alpha}{\sum_{i=1}^{h} \vec{d}_{\sigma(i)}+\beta} $ is maximized.\\

For each $i\in \mathcal{E}$, consider the vector $\vec{a}_i=(a_{i1}, a_{i2}, \hdots, a_{in})$ where
\begin{equation*}
a_{ij} =
\begin{cases}
1 &\text{if $i\in S_j$}\\
0 &\text{otherwise.}
\end{cases}
\end{equation*}
and $\matr{A}= (a_{ij})_{m \times n}$ be an $m \times n$ matrix. Also, consider the decision variables $x_1,x_2,\ldots ,x_n$ where

\begin{equation*}
x_j =
\begin{cases}
1 &\text{if $j$ is in the pack}\\
0 &\text{otherwise.}
\end{cases}
\end{equation*}

The vector of decision variables is represented as $ \vec{x} =  (x_1, \hdots, x_n )^T$ and $\vec{1}$ is column vector of size $n$ where all entries equal to 1. Then the LSPP and LFSPP can be formulated respectively as  0-1 integer programs

\begin{align}
\nonumber \mbox{LSPP:\hspace{1cm}Maximize~ } & \vec{c}\vec{x}\\
\label{eeq5} \mbox{Subject to } & \matr{A}\vec{x} \leq \vec{1}\\
\label{eeq6} &\vec{x} \in \{0,1\}^n
\end{align}
and
\begin{align}
\nonumber \mbox{LFSPP:\hspace{1cm}Maximize~ } & \dfrac{\vec{c}\vec{x} + \alpha}{\vec{d}\vec{x}+\beta}\\
\label{eeq7} \mbox{Subject to } & \matr{A}\vec{x} \leq \vec{1}\\
\label{eeq8} &\vec{x} \in \{0,1\}^n
\end{align}

It is assumed that $\vec{c} \geq \vec{0}$, $\alpha = 0$ and $\beta $ is a scalar such that $\vec{d}\vec{x} + \beta > 0$. Throughout the paper we will assume that $\vec{d}\vec{x} + \beta > 0$ for any feasible solution of LFSPP. \\

The continuous relaxations of LSPP and LFSPP, denoted respectively by LSPP(C) and LFSPP(C),  are obtained by replacing the constraint set $\vec{x} \in \{0,1\}^n$ by $\vec{x} \ge \vec{0}$, respectively in LSPP and LFSPP.\\

The family of feasible solutions of both LSPP and LFSPP is given by $S = \{ \vec{x}  | A\vec{x}\le \vec{1}, \vec{x} \in \{0,1\}^n \}$ and the family of feasible solutions for their continuous relaxations is given by $\bar{S} = \{ \vec{x}  | A\vec{x}\le \vec{1}, \vec{x} \ge \vec{0} \}$.  \\

Following are some definitions given in \cite{gupta1}.  A solution $\vec{x} \in S$ which satisfies (\ref{eeq7}) and  (\ref{eeq8})  is said to be a \textit{pack solution}. For any \textit{pack} $H$, a column of $\matr{A}$ corresponding to  $j \in G $ is said to be redundant if $H + \{j\}$ is also a \textit{pack}.  If a \textit{pack} corresponds to  one or more redundant columns, it is called a \textit{redundant pack}. A \textit{pack} $ H^*$ is said to be a \textit{prime pack}, if none of the columns corresponding to $j^* \in  G$ is \textit{redundant}. A solution corresponding to the \textit{prime pack} is called a \textit{prime packing solution}.\\

The linear fractional set covering problem (LFSCP) is obtained by replacing (\ref{eeq7}) with
\begin{equation}
\matr{A}\vec{x} \geq \vec{1} \label{lfscp}
 \end{equation}

 and  changing the problem from  maximization to minimization problem in the  0-1 integer formulation of  LFSPP. Any $\vec{x} \in \{0,1\}^n$ satisfying (\ref{lfscp}) is called  a \textit{cover solution}.\\

 Any $\vec{x}\in S$ is called a \textit{cover solution} and an optimal solution to the underlying problem LFSCP  is called an \textit{optimal cover solution}. Note that each cover solution corresponds to a cover and viceversa. A \textit{cover} $P$ is said to be redundant if $P - \{j\}$  with $j \in P$ is also a \textit{cover}. A cover which is not \textit{redundant} is called a \textit{prime cover}.
The incidence vector $\vec{x}$ corresponds to \textit{prime cover} is called a \textit{prime cover solution}.\\

For the linear fractional set covering problem, Arora, Puri, and Swarup  \cite{aro} proved that every optimal cover is a prime cover, if $c_{j'} s$ and $d_{j'} s$ satisfy certain conditions.\\

Gupta and Saxena \cite{gupta1}  claimed an extension of the above result to LFSPP,  assuming $\vec{c} \ge \vec{0}, \vec{d} \ge \vec{0}$ and $\vec{d}\vec{x} + \beta > 0$. More precisely, they claimed:

\begin{thm}(Theorem 2 of \cite{gupta1})\label{thm1}
 If the objective function in LFSPP has finite value then, there exists a prime pack solution where this value is attained.
\end{thm}

This result is not true as established by the following example. Let

\[{\begin{array}{cccc}

\beta = 2,

&

\vec{c}=(1,2,5),

&

 \vec{d}=(4,4,6), \mbox{ and }

&

  \matr{A}=
  \begin{pmatrix}
   1  & 1 & 0 \\
1 & 0  & 1 \\
  \end{pmatrix}

  \end{array} }
\]
 For the LFSPP  with $\matr{A},\vec{c}$, $\vec{d}$ and $\beta$ defined as above, it can be verified that  $ \vec{x}^*= (0,0,1)^T$ is an optimal solution  with the objective function value $\dfrac{5}{8}= 0.625$. The optimal \textit{pack} corresponding to $\vec{x}^*$ is $H^*= \{3\}$ which is a \textit{redundant pack} since $ H^* + \{2\} = \{2,3 \}$ is also a \textit{pack}. All other pack solutions and their respective objective function values are  listed below:
\begin{equation*}
\begin{aligned}
& \vec{x}^1 = (1,0,0)^T \mbox{ prime pack solution} \quad &f(\vec{x}^1) = \frac{1}{6} = 0.166\overline{66}\\
&  \vec{x}^2 = (0,1,1)^T  \mbox{ prime pack solution } \quad &f(\vec{x}^2) = \frac{7}{12} = 0.58\overline{33} \\
 &   \vec{x}^3 = (0,0,0)^T \mbox{ redundant pack solution }  &f(\vec{x}^3) = 0 \quad \quad  \quad \quad \quad \\
&  \vec{x}^4 = (0,1,0)^T \mbox{ redundant pack solution }  \quad & f(\vec{x}^4) = \frac{1}{3}=0.33\overline{33}\\
\end{aligned}
\end{equation*}

None of these corresponds to an optimal solution for LFSPP. In particular, no prime pack solution is optimal for the instances of LFSPP constructed above, which contradicts Theorem \ref{thm1}. However, a variation of the Theorem \ref{thm1} can be proved as noted below.

\begin{thm} There always exists a prime pack optimal solution for LFSPP if\\
  (a) $\vec{c} > \vec{0}  $ and $\vec{d}  < \vec{0}$  or \\
  (b) any ratio of the partial sums of $c_i$'s or $d_i$'s is greater than the value of the objective function at the pack solution and also the partial sum of $d_i$'s is positive.
\end{thm}

\begin{proof}

Let us assume that there exists an optimal pack $H_1$ for LFSPP which is a  redundant optimal pack of LFSPP. Since $H_1$ is a redundant pack, a prime pack $H_2$ can be derived from $H_1$  by adding  redundant columns of $H_1$. \\

The objective function value of LFSPP for  $H_1$ and $H_2$ are $Z_{H_1}$ and $Z_{H_2}$ respectively:
\begin{equation*}
Z_{H_1} = \dfrac{\sum_{j\in H_1} c_j  }{\sum_{j\in H_1} d_j+\beta} \text{ and } ~ Z_{H_2} = \dfrac{\sum_{j\in H_1} c_j  + \sum_{j\in \{H_2 - H_1\}} c_j  }{\sum_{j\in H_1} d_j  +\sum_{j\in \{H_2 - H_1\}} d_j   +\beta}
\end{equation*}


\begin{itemize}
\item[(a)]

Since $\vec{c} > \vec{0}  $ and $\vec{d}  < \vec{0}$, therefore,

  \begin{equation}
 \label{e1} \sum_{j\in H_1} c_j < \sum_{j\in H_1} c_j  + \sum_{j\in \{H_2 - H_1\}} c_j
  \end{equation}
  and
  \begin{equation}
  \label{e2} \sum_{j\in H_1} d_j+\beta > \sum_{j\in H_1} d_j +\sum_{j\in \{H_2 - H_1\}} d_j   +\beta
  \end{equation}

  Since denominator of the objective function value is always positive,  dividing inequality (\ref{e1}) by (\ref{e2}), we get:

  \begin{equation*}
   \dfrac{\sum_{j\in H_1} c_j }{\sum_{j\in H_1} d_j+\beta} < \dfrac{\sum_{j\in H_1} c_j  + \sum_{j\in \{H_2 - H_1\}} c_j }{\sum_{j\in H_1} d_j  +\sum_{j\in \{H_2 - H_1\}} d_j   +\beta}
  \end{equation*}

  which gives

  \begin{equation*}
     Z_{H_1} < Z_{H_2}.
    \end{equation*}

   This shows that $H_2$ is an optimal pack for LFSPP instead of $H_1$,  contradicts the optimality of $H_1$, therefore $H_1$ is a prime pack of LFSPP . This completes the proof of part (a).\\

   Note: In this case if we relax  $\vec{c} > \vec{0}$, then above theorem is no longer true.\\

\item[(b)]

    It is given that any ratio of the partial sums of $c_i$'s or $d_i$'s is greater than the value of the objective function at the pack solution and also the partial sum of $d_i$'s is positive and denominator of the objective function is positive for any feasible solution of LFSPP, therefore:

    \begin{eqnarray*}
      \dfrac{ \sum_{j\in \{H_2 - H_1\}} c_j }{\sum_{j\in \{H_2 - H_1\}} d_j } && >  Z_{H_1}, \text{ which implies }\\
      \dfrac{ \sum_{j\in \{H_2 - H_1\}} c_j }{\sum_{j\in \{H_2 - H_1\}} d_j } && >  \dfrac{\sum_{j\in H_1} c_j }{\sum_{j\in H_1} d_j +\beta}, \text{ cross multiplication gives } \\
      (\sum_{j\in \{H_2 - H_1\}} c_j )(\sum_{j\in H_1} d_j +\beta) && > (\sum_{j\in \{H_2 - H_1\}} d_j )( \sum_{j\in H_1} c_j )\\
    \end{eqnarray*}

    adding both sides $(\sum_{j\in H_1} c_j)( \sum_{j\in H_1} d_j +\beta) $  will give

    \begin{eqnarray*}
   &&  (\sum_{j\in H_1} c_j)( \sum_{j\in H_1} d_j +\beta) +  (\sum_{j\in \{H_2 - H_1\}} c_j )(\sum_{j\in H_1} d_j +\beta)  >  \\
       &&   (\sum_{j\in H_1} c_j)( \sum_{j\in H_1} d_j +\beta) + (\sum_{j\in \{H_2 - H_1\}} d_j )( \sum_{j\in H_1} c_j )  \\
    \end{eqnarray*}

    after simplifying, we get

    \begin{equation*}
    ( \sum_{j\in H_1} d_j +\beta) (\sum_{j\in H_1} c_j+ \sum_{j\in \{H_2 - H_1\}} c_j )    >  (\sum_{j\in H_1} c_j)( \sum_{j\in H_1} d_j +\beta+ \sum_{j\in \{H_2 - H_1\}} d_j ) \\
    \end{equation*}
 which is equivalent to
    \begin{equation*}
    \dfrac{(\sum_{j\in H_1} c_j+ \sum_{j\in \{H_2 - H_1\}} c_j }{ \sum_{j\in H_1} d_j +\beta+ \sum_{j\in \{H_2 - H_1\}} d_j ) }   >  \dfrac{(\sum_{j\in H_1} c_j)}{( \sum_{j\in H_1} d_j+\beta)}.  \\
    \end{equation*}

    which gives
     \begin{equation*}
    Z_{H_2} >  Z_{H_1}.
    \end{equation*}

    Therefore, the  objective function value of any prime pack is always greater than any corresponding redundant pack, therefore, the optimal pack is a prime pack. This completes the proof of part (b).

\end{itemize}
\end{proof}


\begin{lem} \label{lem1} If $c$ and $d$ are integer numbers, $k,l=1,\hdots, n$,~ $ l d + \beta > 0$ and $k + l \leq n$, then

\begin{equation*}
    \dfrac{kc}{kd+ \beta} < \dfrac{kc+lc}{kd+ld+ \beta}, ~ ~   \text{ if } c > 0
\end{equation*}

and
\begin{equation*}
    \dfrac{kc}{kd+ \beta} > \dfrac{kc+lc}{kd+ld+ \beta}, ~ ~  \text{ if } c < 0.
\end{equation*}

\end{lem}
\begin{proof}
For given integer numbers $c$ and $d$, if $c >0$ then
\begin{equation}
 \label{e3}   lc \beta > 0
\end{equation}

add  $(kc)(kd+\beta) + (kc)(ld)$ both sides of the inequality  (\ref{e3}), we get

 \begin{equation}
 \label{e4}  (kc)(kd+\beta) + (kc)(ld) +  lc \beta > (kc)(kd+\beta) + (kc)(ld)
\end{equation}

after rearranging inequality (\ref{e4}) we get

\begin{eqnarray}
 \nonumber  (kc)(kd+\beta) + (lc)(kd+\beta)  &&> (kc)(kd+\beta) + (kc)(ld) \text{ which is same as } \\
 \label{e5} (kd+\beta) (kc+lc) && > (kc)(kd+ld+ \beta)
\end{eqnarray}

since $(kd+\beta) > 0$ and $ (kd+ld+ \beta) > 0$, divide both sides of the inequality  (\ref{e5}) by $(kd+\beta)(kd+ld+ \beta)$ , we get
\begin{equation}
\label{e6} \dfrac{kc}{kd+ \beta} < \dfrac{kc+lc}{kd+ld+ \beta}
\end{equation}

which proves the first part.\\

Now if $c < 0$  then

\begin{equation*}
    \dfrac{kc}{kd+ \beta} > \dfrac{kc+lc}{kd+ld+ \beta},
\end{equation*}
this  can be proved  in a similar manner as we did the first part.

This completes the proof.
\end{proof}

\begin{thm}
If $c_i = c, $ and  $d_i = d, \forall i=1,\hdots,n$, and  $H^*$  is an optimal  pack of LFSPP: (a) if $c > 0$, then $H^*$ is a  largest cardinality prime pack solution of LFSPP, (b) if $c < 0$, then $H^*$ is a  smallest cardinality pack solution of LFSPP.
\end{thm}

\begin{proof}(a)  If $c>0$ :\\
 if $H^*$ is  not a prime pack then we can always add redundant columns for $H^*$ and find a prime pack $H^{**}$ of LFSPP and corresponding objective function values for $H^*$  and $H^{**}$ are following:

  \begin{equation*}
    Z_{H^*} = \dfrac{\sum_{j \in H^*} c }{\sum_{j \in H^*} d+ \beta}  \text{ and } Z_{H^{**}}= \dfrac{\sum_{j \in H^*} c + \sum_{j \in \{H^{**}-H^*\} } c  }{\sum_{j \in H^*} d +\sum_{j \in \{H^{**}-H^*\} } d+ \beta}
\end{equation*}
    using  first part of Lemma \ref{lem1} we can claim that
     \begin{equation*}
    Z_{H^*} <  Z_{H^{**}}
\end{equation*}
    which contradicts the optimality of $H^*$. Therefore, $H^*$ is a prime pack of the given LFSPP.\\

    If among all prime packs of LFSPP $H^*$ is not of  the largest cardinality then there exist a prime cover $H^o$ of LFSPP  such that $|H^*| < |H^o|$ and

    \begin{equation*}
    Z_{H^*} = \dfrac{\sum_{j \in H^*} c }{\sum_{j \in H^*} d+ \beta}  \text{ and }Z_{H^o} = \dfrac{\sum_{j \in H^o} c }{\sum_{j \in H^o} d+ \beta}
\end{equation*}

 since $|H^*| < |H^o|$ , then $\sum_{j \in H^*} c < \sum_{j \in H^o} c$, along with  the first part of Lemma \ref{lem1} we can claim that
\begin{equation*}
    Z_{H^*} <  Z_{H^{o}}
\end{equation*}

which contradicts the optimality of $H^*$. Therefore, $H^*$ if the prime pack of the largest cardinality among all prime packs of LFSPP. \\

(b)   If $c< 0$ : \\

     If among all  packs of LFSPP $H^*$ is not of the smallest  cardinality then there exist a pack $H^o$ of LFSPP  such that $|H^*| > |H^o|$ and

    \begin{equation*}
    Z_{H^*} = \dfrac{\sum_{j \in H^*} c }{\sum_{j \in H^*} d+ \beta}  \text{ and }Z_{H^o} = \dfrac{\sum_{j \in H^o} c }{\sum_{j \in H^o} d+ \beta}
\end{equation*}

but since $|H^*| > |H^o|$ , then $\sum_{j \in H^*} c < \sum_{j \in H^o} c$, along with  the second part of Lemma \ref{lem1} we can claim that
\begin{equation*}
    Z_{H^*} <  Z_{H^{o}}
\end{equation*}

which  contradicts the optimality of $H^*$. Therefore, $H^*$ is the  pack of LFSPP of the smallest cardinality. This completes the proof.

\end{proof}


\section{Conclusion}

I would  like to explore non-linear set packing problems in future research. I would like to thank Prof. Abraham P. Punnen for his valuable suggestions during the preparation of this note.


\begin{thebibliography}{1}


\bibitem{aro} Arora, S.R.,  Swaroop, Kanti, and Puri, M.C., The set covering problem with linear fractional functional, \textit{Indian Journal of Pure and Applied Mathematics}, 8, (1977) 578-588.

\bibitem{aro1} Arora, S.R., and Puri, M.C.,  Enumeration Technique for the Set Covering Problem with Linear Fractional Functional as its Objective Functions, \textit{ZAMM - Journal of Applied Mathematics and Mechanics}, 57, (1977) 181-186.

  \bibitem{baz}   Bazaraa, Mokhtar S.; Goode, Jamie J, A cutting-plane algorithm for the quadratic set-covering problem, \textit{Operations Research}, 23, (1975) 150 - 158.

      \bibitem{bec} Bector, C. R. and Bhatt. S. K., A linearization technique for solving integral
linear fractional program, \textit{Proc. fifth Manitoba Conference on Numerical Mathematics}, (1975) 221 - 229.


\bibitem{gar} Garfinkel, M. and Nemhauser. G. L, \textit{Integer Programming}, A Wiley-Interscience
Publication, John Wiley and Sons (1973).


\bibitem{gupta} Gupta, Rashmi, and  Saxena, R. R., Linearization technique for solving quadratic set packing and partitioning problems, \textit{International Journal of Mathematics and Computer Applications Research}, 4, (2014) 9 - 20.

\bibitem{gupta1} Gupta, Rashmi, and  Saxena, R. R., Set packing problem with linear fractional objective function,  \textit{International Journal of Mathematics and Computer Applications Research}, 4, (2014) 9 - 18.



\bibitem{lem} Lemke, C. E., Salkin, H. M. and Spielberg K., Set covering by single
branch enumeration with linear programming sub-problem, \textit{Operations Research}, 19, (1971) 998 - 1022.





  \end{thebibliography}
\end{document}